\documentclass[letterpaper, 10 pt, conference]{ieeeconf}
\IEEEoverridecommandlockouts
\usepackage{graphicx,mathtools,amsfonts,cite,xcolor,pgfplots,algorithm, algorithmic}

\pgfplotsset{compat=1.18}

\DeclareMathOperator{\vspan}{span}
\DeclareMathOperator{\diag}{\mathrm{diag}}

\newcounter{experiment}[section]
\newtheorem{proposition}{Proposition}

\newtheorem{definition}{Definition}

\newtheorem{remark}{Remark}
\newtheorem{lemma}{Lemma}
\newtheorem{theorem}{Theorem}
\newtheorem{corollary}{Corollary}

\allowdisplaybreaks

\title{\LARGE \bf On Dynamic Mode Decomposition of Control-affine Systems}

\author{Moad Abudia, Joel A. Rosenfeld, and Rushikesh Kamalapurkar
    \thanks{This research was supported by the Air Force Office of Scientific Research (AFOSR) under contract numbers FA9550-20-1-0127 and FA9550-21-1-0134, and the National Science Foundation (NSF) under awards  2027976 and 2027999. Any opinions, findings and conclusions or recommendations expressed in this material are those of the author(s) and do not necessarily reflect the views of the sponsoring agencies.}
    \thanks{M. Abudia is with the School of Mechanical and Aerospace Engineering, Oklahoma State University, Stillwater, OK, USA, email: {\tt\small moad.abudia@okstate.edu}}%
    \thanks{J. A. Rosenfeld is with the Department of Mathematics and Statistics, University of South Florida, Tampa, FL, USA, email: {\tt\small rosenfeldj@usf.edu}}%
    \thanks{R. Kamalapurkar is with the Department of Mechanical and Aerospace Engineering, University of Florida, Gainesville, FL, USA, email: {\tt\small rkamalapurkar@ufl.edu}}%
}

\begin{document}
\maketitle
\thispagestyle{empty}
\pagestyle{empty}

\begin{abstract}

    This paper builds on the theoretical foundations for dynamic mode decomposition (DMD) of control-affine dynamical systems by leveraging the theory of vector-valued reproducing kernel Hilbert spaces (RKHSs). Specifically, control Liouville operators and control occupation kernels are used to separate the drift dynamics from the input dynamics. A provably convergent finite-rank estimation of a compact control Liouville operator is obtained, provided sufficiently rich data are available. A matrix representation of the finite-rank operator is used to construct a data-driven representation of its singular values, left singular functions, and right singular functions. The singular value decomposition is used to generate a data-driven model of the control-affine nonlinear system. The developed method generates a model that can be used to predict the trajectories of the system in response to any admissible control input. Numerical experiments are included to demonstrate the efficacy of the developed technique.

\end{abstract}

\section{Introduction}
Dynamic mode decomposition (DMD) is a data analysis method that aims to generate a finite-rank representation of a transfer operator corresponding to a nonlinear dynamical system using time series measurements \cite{SCC.Kutz.Brunton.ea2016,SCC.Budisic.Mohr.ea2012,SCC.Mezic2005,SCC.Korda.Mezic2018}. The time series is expressed as a linear combination of the dynamic modes. The coefficients in the linear combination are given by exponential functions of time. The dynamic modes and the growth rates of the exponential functions are derived from the spectrum of a finite-rank representation of the Koopman operator (or, in the continuous-time case, Koopman generator). In \cite{SCC.Korda.Mezic2018}, it was established that the finite rank representation of the Koopman operator converges, in the strong operator topology (SOT), to the true Koopman operator as the amount of data increases. However, convergence of the spectrum doesn't necessarily follow from convergence in the SOT \cite{SCC.Pedersen2012}. Since the DMD procedure relies on the spectrum to construct a model, the constructed model is not guaranteed to converge to the true dynamical system model that is being identified.

Kernel methods developed by the machine learning community have been adopted for system identification purposes by the control community \cite{pillonetto2014kernel,pillonetto2010new,pillonetto2011new}. The complexity of kernel methods is further investigated in the system identification context in \cite{pillonetto2015tuning}, where it shows that tuning of the hyperparameters is a problem that still persists today.

In \cite{SCC.Abudia.Rosenfeld.ea2024}, Liouville operators (or Koopman generators) were used instead of the Koopman operator, where examples of RKHSs and dynamical systems for which  the  Liouville operators are compact are provided. Furthermore, a finite-rank estimation of the compact Liouville operator is shown to converge in norm to the true Liouville operator. However, \cite{SCC.Abudia.Rosenfeld.ea2024} does not consider controlled systems, which is the focus of this paper.

In \cite{SCC.Proctor.Brunton.ea2016} a DMD routine to represent a general nonlinear system with control as a control-affine linear system. This idea is generalized in \cite{SCC.Korda.Mezic2018a} with extended DMD (eDMD), which provides good predictions but with no spectral convergence guarantees. Additionally, for a general discrete-time, nonlinear dynamical system with control, \cite{SCC.Korda.Mezic2018a} utilizes the shift operator to describe the time evolution of the control signal. Also, in discrete-time, separation of the control input from the state can be achieved via first order approximations \cite{STRASSER20232257}. For continuous-time dynamical systems, the Koopman canonical transform (see \cite{SCC.Surana2016}) is used in \cite{SCC.Goswami.Paley2022} to lift the nonlinear dynamical system and approximate it as a control-affine, bilinear system, called the Koopman bilinear form (KBF). The KBF is then amenable to the design of feedback laws using techniques from optimal control.

More recently, in \cite{strasser2024koopman} a lifted bilinear representation of the nonlinear system is identified and is used to develop a feedback controller that is guaranteed to stabilize the system given sufficiently rich data and an appropriate region of attraction. However, given the condition that the data must be obtained from a system that is activated using every basis of the control space.

In this paper, we present a method to identify a control-affine nonlinear system via a finite-rank operator, which converges (in norm) to the true compact control Liouville operator provided sufficiently rich data are available. This work uses control occupation kernels,  which were introduced in \cite{SCC.Rosenfeld.Kamalapurkar2021}, augmentations of occupation kernels, that incorporate control signal information by leveraging vector valued reproducing kernel Hilbert spaces (vvRKHS). Occupation kernels themselves were introduced in \cite{SCC.Rosenfeld.Russo.ea2024}, where system identification problems are addressed not through numerical differentiation, but rather through integration. Using numerical integration is considerably less sensitive to noise \cite{SCC.Abudia.Channagiri.eainprep} as opposed to numerical differentiation used in popular system identification methods such as SINDYc \cite{SCC.Brunton.Proctor.ea2016}.


The main contribution of this paper is the development of a DMD routine for nonlinear control-affine dynamical systems using a compact control Liouville operator. The finite-rank approximation is provably convergent, in norm, to the true control Liouville operator on Bargmann-Fock spaces restricted to the set of real numbers.

\section{Preliminaries} \label{sec:Preliminaries}
Given a Hilbert space $\mathcal{Y}$ and a set $X$, a vector valued RKHS (vvRKHS), $H$, is a Hilbert space of functions from $X$ to $\mathcal{Y}$, such that for each $v\in \mathcal{Y}$ and $x \in X$, the functional $f \mapsto \langle f(x), v \rangle_{\mathcal{Y}}$ is bounded. Hence for each $x \in X$ and $v \in \mathcal{Y}$, there is a function $K_{x, v} \in H$ such that $\langle f(x), v \rangle_{\mathcal{Y}} = \langle f, K_{x,v} \rangle_{H}$. The mapping $v \mapsto K_{x,v}$ is linear over $\mathcal{Y}$; hence, $K_{x,v}$ may be expressed as an operator over $\mathcal{Y}$ as $K_x v := K_{x,v}$. The operator $K(x,y) := K_y^*K_x$ is called the reproducing kernel operator corresponding to $H$.

In the present context, we define  $\mathcal{Y} = \mathbb{R}^{m+1}$ (viewed as row vectors), $X = \mathbb{R}^n$, and $H$ consists of continuously differentiable functions. Since $\mathcal{Y}$ is the space of row-vectors, the operation on $v \in \mathbb{R}^{m+1}$ by the kernel operator $K_x$ will be expressed as $v K_{x}$. Given two continuous signals, $\theta : [0,T] \to \mathbb{R}^n$ and $u : [0,T] \to \mathbb{R}^m$, the control occupation kernel corresponding to this pair of signals is the unique function, $\Gamma_{\theta,u} \in H$, that represents the bounded functional $h \mapsto \int_{0}^T h(\theta(t)) \begin{pmatrix} 1 \\ u(t) \end{pmatrix} dt$ as $\langle h, \Gamma_{\theta,u}\rangle_H=\int_{0}^T h(\theta(t)) \begin{pmatrix} 1 \\ u(t) \end{pmatrix} dt$. 
via the following proposition. 

\begin{proposition}\label{prop:ftc}
\cite{SCC.Abudia.Channagiri.eainprep} If $H$ is a vvRKHS consisting of $\mathbb{R}^{m+1}$ (row) valued continuous functions over $\mathbb{R}^n$, $\gamma_u: [0,T] \to \mathbb{R}^n$ is a controlled trajectory with continuous controller $u:[0,T]\to\mathbb{R}^m$ satisfying, for all $t\in[0,T]$, $\frac{d^s\gamma_{u}}{dt^s} (t) = f(\gamma_{u}(t)) + g(\gamma_{u}(t))u(t)$, and $\begin{pmatrix} (f)_{j} & (g)_{j} \end{pmatrix} \in H$ for each $j=1,\ldots,n$, then
$
    \langle \begin{pmatrix} (f)_{j} & (g)_{j} \end{pmatrix}, \Gamma^{(s)}_{\theta_u, u} \rangle_H = \left(\gamma_{u}(T) - \sum_{\ell = 0}^{s-1} \frac{T^\ell}{\ell!}\frac{d^{\ell}\gamma_{u}}{dt^{\ell}}(0)\right)_j.
$
\end{proposition}

\begin{proof}
See the proof of 
 Proposition 1 in \cite{SCC.Abudia.Channagiri.eainprep}.
\end{proof}

While Proposition \ref{prop:ftc} follows from a simple application of the fundamental theorem of calculus, it sets the stage for a powerful approximation routine leveraging higher order control occupation kernels. These kernels can be implemented directly using a direct interpolation approach, or they can arise naturally in a regularized regression problem.

\begin{proposition}
\cite{SCC.Abudia.Channagiri.eainprep} Fix $s \in \mathbb{N}$, let $H$ be a vvRKHS with the kernel operator $K_x$. If $\theta$ and $u$ are continuous signals from $[0,T]$ to $\mathbb{R}^n$ and $\mathbb{R}^m$ respectively. Then the control occupation kernel $\Gamma^{(s)}_{\theta,u}(x)$ can be obtain by \begin{equation}\label{eq:evaluate-occupation}
\Gamma^{(s)}_{\theta,u}(x) = \frac{1}{(s-1)!} \int_0^T (T-t)^{s-1} \begin{pmatrix} 1 & u^{\top}(t) \end{pmatrix} K_{\theta(t)}(x)dt.
\end{equation}
\end{proposition}

\begin{proof}
See the proof of 
 Proposition 2 in \cite{SCC.Abudia.Channagiri.eainprep}.
\end{proof}


If variable length trajectories are admitted, then
it can be shown that the span of the occupation kernels corresponding to trajectories $\theta_u$ that result from the application of the control $u$ to the system in \eqref{eq:control-affine-dynamics} is dense in $H$.

\begin{proposition}\label{prop:SystemTrajectoryDensity}
For any order $s$, the span of the set
\[
    A_s\coloneqq\left\{\Gamma_{\gamma_{u,\gamma_0},u}^{(s)}\mid \begin{gathered}
        \gamma_{u,\gamma_0}\in C([0,T_u];\mathbb R^n),\gamma_0\in\mathbb{R}^n,\\
        u\in C([0,T_u];\mathbb R^m),T_u\in[0,T]
    \end{gathered} \right\}
\]
is dense in $H$, where $\gamma_{u,\gamma_0}$ is a solution of \eqref{eq:control-affine-dynamics}, under the control $u$ and starting from $\gamma_0$.
\end{proposition}
\begin{proof}
    Select $\gamma_0$ such that $h(\gamma_0) \neq 0$. Continuity of $h$ and $\gamma_{u,\gamma_0}$ can then be invoked to conclude that for any constant control signal $u(t) = b$ such that $h(\gamma_0) \begin{pmatrix} 1 \\ b \end{pmatrix} \neq 0$, there exists a $T_u > 0$ for which $ \int_{0}^{T_u} h(\gamma_{u,\gamma_0}(t)) \begin{pmatrix} 1 \\ b \end{pmatrix} dt \neq 0$. For example, one can select $T_u = \min\left\{T,\inf_t\left\{h(\gamma_{u,\gamma_0}(t))\begin{pmatrix} 1 \\ b \end{pmatrix}=0\right\}\right\}$. A straightforward extension of the above argument to the higher order case allows us to conclude that for any order $s$ and any nonzero $h$, there exist $\gamma_0$, $u$, and $T_u$ such that $\left\langle h , \Gamma_{\gamma_{u,\gamma_0},u}^{(s)} \right\rangle_H \neq 0$. That is, $H\cap A_s^\perp = \{0\}$, and as a result, $(A_s^\perp)^\perp = H$. Since $(A_s^\perp)^\perp = \overline{\vspan A_s}$, we conclude that $H = \overline{\vspan A_s}$.
\end{proof}
Proposition \ref{prop:SystemTrajectoryDensity} motivates the use of control occupation kernels for system identification. If $H$ is universal, then any continuous function can be approximated, uniformly over any compact set, by a function in $H$, and any function in $H$ can be approximated using a linear combination of sufficiently many control occupation kernels.

\begin{definition}\label{def:totalderivative}
   \cite{rosenfeld2024dynamicmodedecompositioncontrol} The \emph{control Liouville operator with symbol $f,g$} denoted by $A_{f,g}: \mathcal{D}(A_{f,g}) \to H$ is defined $A_{f,g} h(x) :=   \nabla_x h(x) \begin{pmatrix}
        f(x) & g(x)
    \end{pmatrix}$ where the domain $\mathcal{D}(A_{f,g})$ is defined canonically as
    \begin{equation}\label{eq:totalderivative}
        \mathcal{D}(A_{f,g}) = \{ h \in H : A_{f,g} h \in  H\}.
    \end{equation}
\end{definition}

There is a relation between the adjoint of the control Liouville operator and the control occupation kernel, which is presented in the following proposition

\begin{proposition}
(\cite[Proposition 3]{rosenfeld2024dynamicmodedecompositioncontrol})
Suppose that $f$ and $g$ correspond to a control Liouville operator, $A_{f,g}:\mathcal{D}(A_{f,g}) \to H$, and let $u$ be an admissible control signal for the control-affine dynamical system in (\ref{eq:control-affine-dynamics}) with a corresponding controlled trajectory, $\gamma_u$. Then, $\Gamma_{\gamma_u ,u} \in {D}(A^*_{f,g})$ and $A^*_{f,g}\Gamma_{\gamma_u ,u}= K(\cdot,\gamma_u(T))-K(\cdot,\gamma_u(0))$. 
\end{proposition}
\begin{proof}
    See the proof of Proposition 3 in \cite{rosenfeld2024dynamicmodedecompositioncontrol}
\end{proof}

Consider the exponential dot product kernel with parameter $\tilde{\rho}$, defined as $\tilde{K}_{\tilde{\rho}}(x,y) = \exp\left(\frac{x^{\top}y}{\tilde{\rho}}\right)$. In the single variable case, the RKHS of this kernel is the restriction of the Bargmann-Fock space to real numbers, denoted by $F^2_{\tilde{\rho}}\left(\mathbb{R}\right)$. This space consists of the set of functions of the form $h(x) = \sum_{k=0}^\infty a_k x^k$, where the coefficients satisfy $\sum_{k=0}^\infty \left\vert a_k\right\vert^2 \tilde{\rho}^k k! < \infty$, and the norm is given by $ \left\Vert h \right\Vert^{2}_{\tilde{\rho}} = \sum_{k=0}^\infty \left\vert a_k\right\vert^2 \tilde{\rho}^k k! $. Note that the set of polynomials in $x$ is a subset of $F^2_{\tilde{\rho}}\left(\mathbb{R}\right)$. Extension of this definition to the multivariable case yields the space $F^2_{\tilde{\rho}}\left(\mathbb{R}^n\right)$ where the collection of monomials, $x^{\alpha} \frac{\tilde{\rho}^{|\alpha|}}{\sqrt{\alpha!}}$, with multi-indices $\alpha \in \mathbb{N}^n$ forms an orthonormal basis.
In this setting, provided  $\tilde{\rho}_2 < \tilde{\rho}_1$, differential operators from $F^2_{\tilde{\rho}_1}(\mathbb{R}^n)$ to $F^2_{\tilde{\rho}_2}(\mathbb{R}^n)$ can be shown to be compact \cite[see Proposition 8]{SCC.Rosenfeld.Kamalapurkar2021}.

\section{Problem Statement}\label{sec:ProbS}
The objective in this manuscript is to learn an unknown control affine system from observed control signals and controlled trajectories, $\{ u_j : [0,T_j] \to \mathbb{R}^m \}_{i=1}^M$ and $\{ \gamma_{u_j}:[0,T_j] \to \mathbb{R}^n\}_{j=1}^M$, respectively, satisfying
\begin{equation} \label{eq:control-affine-dynamics}
    \dot{x} = F(x,u) = f(x) + g(x) u,
\end{equation}
where $f : \mathbb{R}^n \to \mathbb{R}^n$ is the drift function and $g:\mathbb{R}^n \to \mathbb{R}^{n\times m}$ is the control effectiveness matrix.

Systems of the form \eqref{eq:control-affine-dynamics} encompass linear systems and Euler-Lagrange models with invertible inertia matrices, and hence, represent a wide class of physical plants, including but not limited to robotic manipulators and autonomous ground, aerial, and underwater vehicles.

In order to facilitate the description of the controlled dynamical system in terms of operators, the vvRKHS framework from section \ref{sec:Preliminaries} is utilized in this paper. 
\begin{definition}
    Given a compact control Liouville operator $A_{f,g}:\tilde{H}_{d}\to H$, where $\tilde{H}_{d}$ is a scaler valued RKHS. The tuples $\{(\sigma_i,\phi_i,\psi_i)\}_{i=1}^\infty$, with $\sigma_i\in \mathbb{R}^n$, $\phi_i\in \tilde{H}_{d}$, and $\psi_i \in H$, are singular values, left singular vectors, and right singular vectors of $A_{f,g}$, respectively, if $\forall h\in\vspan{d}$, $ A_{f,g} h = \sum_{i=1}^\infty \sigma_i \psi_i \left\langle h,\phi_i\right\rangle_{\tilde{H}_{d}}$. \label{eq:infinite_spectral_reconstruction_svd}
\end{definition}
Let $h_{\mathrm{id}}:\mathbb{R}^n \to \mathbb{R}^n$ be the identity function. Given singular values, left singular vectors, and right singular vectors of $A_{f,g}$ and a control signal $u$, the dynamics of the system can be expressed as
\begin{multline}
    \dot{x} = \begin{pmatrix}A_{f,g} (h_{\mathrm{id}})_1\\\vdots\\A_{f,g} (h_{\mathrm{id}})_n\end{pmatrix} \begin{pmatrix}
        1\\u
    \end{pmatrix} \\
    = \begin{pmatrix}
    \sum_{i=1}^\infty \sigma_i \psi_i \left\langle (h_{\mathrm{id}})_1,\phi_i\right\rangle_{\tilde{H}_{d}}\\\vdots\\ \sum_{i=1}^\infty \sigma_i \psi_i \left\langle (h_{\mathrm{id}})_M,\phi_i\right\rangle_{\tilde{H}_{d}}
    \end{pmatrix}\begin{pmatrix}
        1\\u
    \end{pmatrix}
\end{multline}
The objective of this work is to generate a provably convergent finite truncation of the above model. To that end, we derive a finite rank representation of the control Liouville operator.

\section{Finite-rank Representation of the Control Liouville Operator} \label{sec:Finite-rank Representation}
To facilitate computation, an explicit finite-rank representation of $A_{f,g}$ is needed to determine the dynamic modes of the resultant system. In the following, finite collections of linearly independent vectors, $d^M$ and $\beta^M$ are selected to establish the needed finite-rank representation. Since the adjoint of $ A_{f,g} $ maps control occupation kernels to kernel differences (\cite[Proposition 3]{rosenfeld2024dynamicmodedecompositioncontrol}), the span of the collection of kernel differences
\begin{equation}
    d^M = \left\{K_d(\cdot,\gamma_{u_i}(T_i)) - K_d(\cdot,\gamma_{u_i}(0))\right\}_{i=1}^M\subset \tilde{H}_{d}
\end{equation}
is selected to be the domain of  $ A_{f,g} $. The corresponding Gram matrix is denoted by $G_{d^M} = \left(\left\langle d_i,d_j\right\rangle_{\tilde{H}_d}\right)_{i,j=1}^M$. The output of $A_{f,g}$ is projected onto the span of the control occupation kernels
\begin{equation}
    \beta^M = \left\{\Gamma_{\gamma_{u_i},u_i}\right\}_{i=1}^M\subset H.
\end{equation}
The corresponding Gram matrix is denoted by  $G_{\beta^M} = \left(\left\langle\beta_i,\beta_j\right\rangle_{H}\right)_{i,j=1}^M$.

A rank-$M$ (or less) representation of the operator $A_{f,g}$ is then given by  $P_{\beta^M} A_{f,g} P_{d^M}:\tilde{H}_d\to\vspan \beta^M$, where $P_{d^M}$ and $P_{\beta^M}$ denote projection operators onto $\vspan{d^M}$ and $\vspan{\beta^M}$, respectively.

Under the compactness assumptions and given rich enough data so that the spans of $ \{d_i\}_{i=1}^\infty $ and $\{\beta_i\}_{i=1}^\infty $ are dense in $\tilde{H}_{d}$ and $H$, respectively, the sequence of finite-rank operators $\{P_{\beta^M} A_{f,g} P_{d^M}\}_{M=1}^{\infty}$ can be shown to converge, in norm, to $A_{f,g}$. To facilitate the proof of convergence, we recall the following result from \cite{SCC.Rosenfeld.Kamalapurkar.ea2022}.
\begin{lemma}\label{lem:technical-lemma}
    Let $H$ and $G$ be RKHSs defined on $X\subset \mathbb{R}^n$ and let $A_N:H\to G$ be a finite-rank operator with rank $N$. If the spans of $\{d_i\}_{i=1}^{\infty}$ and $\{\beta_i\}_{i=1}^{\infty}$ are dense in $H$  and $G$, respectively, then for all $\epsilon > 0 $, there exists $M(N)\in\mathbb{N}$ such that for all $i\geq M(N)$ and $h\in H$, $\left\Vert A_N h - A_N P_{d^i} h\right\Vert_{G} \leq \epsilon \left\Vert h \right\Vert_H$ and $\left\Vert A_N h - P_{\beta^i} A_N h\right\Vert_{G} \leq \epsilon \left\Vert h \right\Vert_H$.
\end{lemma}
\begin{proof}
    See the proof of \cite[Theorem 2]{SCC.Rosenfeld.Kamalapurkar.ea2022}.
\end{proof}
The convergence result for Liouville operators on Bargmann-Fock spaces restricted to the set of real numbers, which will be used for $H$ and $\tilde{H}_d$, follows from the following more general result.
\begin{proposition}\label{prop:convergence}
    If $A:\tilde{H}_d\to H$ is a compact operator and the spans of $ \{d_i\}_{i=1}^\infty $ and $\{\beta_i\}_{i=1}^\infty $ are dense in $\tilde{H}_{d}$ and $H$, respectively, then $\lim_{M\to\infty}\left\Vert A - P_{\beta^M} A P_{d^M} \right\Vert_{\tilde{H}_d}^{H} = 0$, where $\left\Vert\cdot\right\Vert_{\tilde{H}_d}^{H}$ denotes the operator norm of operators from $\tilde{H}_d$ to $H$.
\end{proposition}
\begin{proof}
    Let $\{A_N\}_{N=1}^\infty$ be a sequence of rank-$N$ operators converging, in norm, to $A$. For an arbitrary $h\in \tilde{H}_d$, 
\begin{multline*}
    \left\Vert Ah - P_{\beta^M} AP_{d^M} h\right\Vert_{H} \\
    \leq \left\Vert Ah - A_{N}h  \right\Vert_{H} 
    + \left\Vert A_{N}h - A_{N}P_{d^M}h \right\Vert_{H} \\
    +\left\Vert A_{N}P_{d^M}h - P_{\beta^M}A_{N}P_{d^M}h \right\Vert_{H}  \\
    +\left\Vert P_{\beta^M}A_{N}P_{d^M}h - P_{\beta^M} A P_{d^M}h\right\Vert_{H}.
\end{multline*}
Using the fact that $A_N$ and $P_{\beta^M}A_{N}P_{d^M}$ are finite-rank operators, Lemma \ref{lem:technical-lemma}, can be used to conclude that for all $\epsilon > 0 $, there exists $M(N)\in\mathbb{N}$ such that for all $i\geq M(N)$
\begin{multline*}
    \left\Vert Ah - P_{\beta^i} AP_{d^i} h\right\Vert_{H} 
    \leq \left\Vert Ah - A_{N}h  \right\Vert_{H}\\
    + 2\epsilon\left\Vert h \right\Vert_{\tilde{H}_d}
    +\left\Vert A_{N}P_{d^i}h - A P_{d^i}h\right\Vert_{H}.
\end{multline*}
Since $A_{N}$ converges to $A$ in norm, given $\epsilon > 0$, there exists $N\in \mathbb{N}$ such that for all $j\geq N$, and  $g\in \tilde{H}_d$ $\left\Vert Ag - A_{j}g  \right\Vert_{H} \leq \epsilon \left\Vert g  \right\Vert_{\tilde{H}_d} $. Thus, for all  $j\geq N$ and  $i\geq M(j)$, $ \left\Vert Ah - P_{\beta^i} AP_{d^i} h\right\Vert_{H}
    \leq 4\epsilon\left\Vert h \right\Vert_{\tilde{H}_d}$.
\end{proof}
The convergence result can then be stated as follows.
\begin{theorem}\label{thm:norm-convergence}
    Let $\rho_d\in\mathbb{R}$, $\varrho_d\in\mathbb{R}$, and $\rho=\begin{bmatrix}
        \rho_1&\ldots&\rho_{m+1}
    \end{bmatrix}^{\top}\in\mathbb{R}^{m+1}$ be parameters such that $ \rho_i < \varrho_d $, and $\varrho_d < \rho_d $ for $i = 1,\ldots,m+1$. Let $\tilde{H}_d = F^2_{\tilde{\rho}_d}(\mathbb{R}^n)$, $\tilde{G}_d = F^2_{\tilde{\varrho}_d}(\mathbb{R}^n)$, and $H = F^2_{\rho}(\mathbb{R}^n)$. If $f$ and $g$ are component-wise polynomial, and if the spans of the collections $ \{d_i\}_{i=1}^\infty $ and $\{\beta_i\}_{i=1}^\infty $ are dense in $\tilde{H}_{d}$ and $H$, respectively, then $\lim_{M\to\infty}\left\Vert A_{f,g} - P_{\beta^M} A_{f,g} P_{d^M} \right\Vert_{\tilde{H}_d}^{H} = 0$.
\end{theorem}
\begin{proof}
    \cite[Propositions 7, 8, and 9]{rosenfeld2024dynamicmodedecompositioncontrol} imply that $A_{f,g}$ is compact and hence, the theorem follows from Proposition \ref{prop:convergence}. 
\end{proof}
\section{Matrix Representation of the Finite-rank Operator} \label{sec:Matrix Representation}
 To formulate a matrix representation of the finite-rank operator $ P_{\beta^M} A_{f,g}P_{d^M} $, the operator is restricted to $\vspan d^M$ to yield the operator $ P_{\beta^M} A_{f,g}|_{d^M}:\vspan{d^M}\to\vspan{\beta^M}$. For brevity of exposition, the superscript $M$ is suppressed hereafter and $d$ and $\beta$ are interpreted as $M-$dimensional vectors. 
\begin{proposition}
If $ h = \delta^\top d\in \vspan{d}$ is a function with coefficients $\delta\in\mathbb{R}^M$ and if $g = P_\beta A_{f,g} h$, then $g = a^{\top}\beta$, where $a = G_\beta^{+} G_d \delta$ and $(\cdot)^{+}$ denotes the Moore-Penrose pseudoinverse.
\end{proposition}
\begin{proof}
    \cite[Propositions 3 and 6]{rosenfeld2024dynamicmodedecompositioncontrol} imply that that for all $j=1,\cdots,M$, $ A_{f,g}^* \beta_j = d_j$. Note that since $g$  is a projection of $ A_{f,g} h$ onto $\vspan \beta$, $g = a^{\top}\beta$ for any $a$ that solves
    \begin{equation} \label{eq:alpha_projection}
        G_\beta a = \begin{bmatrix}
        \left\langle A_{f,g}h,\beta_1\right\rangle_{H}\\\vdots\\\left\langle A_{f,g}h,\beta_M\right\rangle_{H}
        \end{bmatrix} = \begin{bmatrix}
            \left\langle h,A_{f,g}^*\beta_1\right\rangle_{H}\\\vdots\\\left\langle  h,A_{f,g}^*\beta_M\right\rangle_{H}
        \end{bmatrix} 
    \end{equation}
    Using the adjoint relationship,
    \begin{equation}\label{eq:beta_projection}
        G_\beta a = \begin{bmatrix}
            \left\langle h,d_1\right\rangle_{H}\\\vdots\\\left\langle  h,d_M\right\rangle_{H}
        \end{bmatrix} = \begin{bmatrix}
            \left\langle \delta^\top d,d_1\right\rangle_{H}\\\vdots\\\left\langle  \delta^\top d,d_M\right\rangle_{H}
        \end{bmatrix} = G_d\delta
    \end{equation}
    Selecting the solution 
    \begin{equation}
        a = G_\beta^+G_d\delta\label{eq:b_matrix}
    \end{equation}
    of \eqref{eq:beta_projection}, a matrix representation $[P_\beta A_{f,g}]_d^\beta$ of the operator $P_\beta A_{f,g}|_d$ is obtained as $G_\beta^{+} G_d$.
\end{proof}
Note that matrix representations are generally not unique. Different representations may be obtained by selecting different solutions of \eqref{eq:alpha_projection} and \eqref{eq:b_matrix}. In the case where the Gram matrix $G_{\beta}$ is nonsingular, equation \eqref{eq:alpha_projection} has a unique solutions, resulting in the unique matrix representation $G_\beta^{-1} G_d$.

In the following section, the matrix representation $[P_\beta A_{f,g}]_d^\beta$ is used to construct a data-driven representation of the singular values and the left and right singular functions of $P_\beta A_{f,g}\vert_d$. 

\section{Singular Functions of the Finite-rank Operator} \label{sec:SVD}
The following proposition states that the SVD of $P_\beta A_{f,g}|_d$ can be computed using matrices in the matrix representation $[P_\beta A_{f,g}]_d^\beta$ derived in the previous section.
\begin{proposition}
    If $ (W,\Sigma,V) $ is the SVD of $G_\beta^{+}$ with $W = \begin{bmatrix} w_1,&\ldots,&w_M \end{bmatrix}$, $V = \begin{bmatrix} v_1,&\ldots,&v_M \end{bmatrix}$, and $\Sigma = \diag\left(\begin{bmatrix} \sigma_1,&\ldots,&\sigma_M \end{bmatrix}\right)$, then for all $i=1,\ldots,M$, $\sigma_i$ are singular values of $P_\beta A_{f,g}|_d$ with left singular functions $\phi_i := v_i^\top d$ and right singular functions $\psi_i := w_i^\top \beta$.
\end{proposition}
\begin{proof}
    Let $\phi_i = v_i^\top d$ and $\psi_i = w_i^\top \beta$ and $h = \delta^\top d$. Then, 
\begin{multline*}
    P_\beta A_{f,g} h = \sum_{i=1}^M \sigma_i \psi_i \left\langle h,\phi_i\right\rangle_{\tilde{H}_{d}}\\
    \iff P_\beta A_{f,g} \delta^\top d = \sum_{i=1}^M \sigma_i w_i^\top \beta \left\langle \delta^\top d, v_i^\top d\right\rangle_{\tilde{H}_{d}}
\end{multline*}
Using the finite-rank representation, the collection $\{(\sigma_i,\phi_i,\psi_i)\}_{i=1}^M$, is an SVD of $P_\beta A_{f,g}|_d$, if for all $\delta\in\mathbb{R}^M$,
\begin{equation}\label{eq:suff_cond_SVD}
    \left(G_\beta^{+} G_d \delta\right)^\top \beta = \left(\sum_{i=1}^M \sigma_i \left\langle \delta^\top d, v_i^\top d\right\rangle_{\tilde{H}_{d}} w_i^\top \right)\beta.
\end{equation}
Simple matrix manipulations yield the chain of implications
\begin{gather*}
    {\thickmuskip=0mu\thinmuskip=0mu\medmuskip=0mu\eqref{eq:suff_cond_SVD}\impliedby \forall \delta\in\mathbb{R}^M, G_\beta^{+} G_d \delta = \sum_{i=1}^M \sigma_i \left\langle \delta^\top d, v_i^\top d\right\rangle_{\tilde{H}_{d}} w_i}\\
    \iff \forall \delta\in\mathbb{R}^M,G_\beta^{+} G_d \delta 
    = \sum_{i=1}^M \sigma_i  \left(w_i v_i^\top G_d\right) \delta\\
    \impliedby G_\beta^{+} G_d = \sum_{i=1}^M \sigma_i  \left(w_i v_i^\top \right)G_d\\
    \impliedby G_\beta^{+} = \sum_{i=1}^M \sigma_i w_i v_i^\top = W\Sigma V^\top,
\end{gather*}
which proves the proposition.
\begin{remark}
    The standard usage of the term SVD refers to a decomposition using orthonormal left and right singular vectors. The left singular functions $\{\phi_i\}_{i=1}^{M}$ and right singular functions $\{\psi_i\}_{i=1}^{M}$ are not necessarily orthonormal. Therefore, the decomposition in the previous proposition can not be properly called and SVD of the finite-rank operator $P_\beta A_{f,g}|_d$.
\end{remark}
\end{proof}
In the following section, the singular values and the left and right singular vectors are used, along with a finite truncation of \eqref{eq:infinite_spectral_reconstruction_svd} to generate a data-driven model.

\section{The SCLDMD Algorithm} \label{sec:SCLDMD}
Motivated by \eqref{eq:infinite_spectral_reconstruction_svd}, assuming that $h_{\mathrm{id},j} \in \tilde{H}_{d}$ for $j=1,\cdots,n$, the system dynamics are approximated using the finite-rank representation, with rank at most $M$, as 
\[
    \dot{x} \approx \hat{F}_{M} (x,u) := [P_{\beta} A_{f,g} P_{d} h_{\mathrm{id}}] (x) \begin{pmatrix}
        1\\u
    \end{pmatrix},
\]    
where $P_{\beta} A_{f,g} P_{d} h_{\mathrm{id}}$ denotes row-wise operation of the operator $P_{\beta} A_{f,g}$ on the function $P_{d} h_{\mathrm{id}}$. 
Using the definition of singular values and singular functions of $P_{\beta} A_{f,g}\mid_d$,
\begin{equation}
    \dot{x} \approx 
    \sum_{i=1}^M \sigma_i \xi_i w_i^\top \beta (x) \begin{pmatrix}
        1\\u
    \end{pmatrix} = \xi\Sigma W^\top \beta (x) \begin{pmatrix}
        1\\u
    \end{pmatrix},
\end{equation}
where $\xi_i \coloneqq \left\langle P_d h_{\mathrm{id}},\phi_i\right\rangle_{\tilde{H}_{d}}$ and $\xi := \begin{bmatrix}
    \xi_1,&\ldots,&\xi_M
\end{bmatrix}$.

The modes $\xi$ can be computed using $\phi_i = v_i^\top d$ as
\begin{gather*}
    \xi = \begin{bmatrix}
    \left\langle P_d h_{\mathrm{id},1},v_1^\top d\right\rangle_{\tilde{H}_{d}},&\ldots,&\left\langle P_d h_{\mathrm{id},1},v_M^\top d\right\rangle_{\tilde{H}_{d}}\\
    \vdots & \ddots & \vdots\\
\left\langle P_d h_{\mathrm{id},n},v_1^\top d\right\rangle_{\tilde{H}_{d}},&\ldots,&\left\langle P_d h_{\mathrm{id},n},v_M^\top d\right\rangle_{\tilde{H}_{d}}
\end{bmatrix}\\
= \begin{bmatrix}
    \left\langle \delta_1^\top d,d_1\right\rangle_{\tilde{H}_{d}},&\ldots,&\left\langle \delta_1^\top d,d_M\right\rangle_{\tilde{H}_{d}}\\
    \vdots & \ddots & \vdots\\
\left\langle \delta_n^\top d,d_1\right\rangle_{\tilde{H}_{d}},&\ldots,&\left\langle \delta_n^\top d,d_M\right\rangle_{\tilde{H}_{d}}
\end{bmatrix} V
= \delta^\top G_d V,
\end{gather*}
where $\delta \coloneqq  \begin{bmatrix} \delta_1, &\ldots, &\delta_n \end{bmatrix}$. Using the reproducing property of the reproducing kernel of $\tilde{H}_d$, the coefficients $\delta_i$ in the projection of $h_{\mathrm{id},i}$ onto $d$ satisfy
{\thinmuskip=0mu \thickmuskip=0mu \medmuskip=0mu \begin{equation*}
    G_d \delta_i = \begin{bmatrix}
        \left\langle\left(h_{\mathrm{id}}\right)_i,d_1\right\rangle_{\tilde{H}_{d}}\\ \vdots \\ \left\langle\left(h_{\mathrm{id}}\right)_i,d_M\right\rangle_{\tilde{H}_{d}}
    \end{bmatrix} = \begin{bmatrix}
        \left(\gamma_{u_1}(T_1)\right)_i - \left(\gamma_{u_1}(0)\right)_i\\ \vdots \\ \left(\gamma_{u_M}(T_M)\right)_i - \left(\gamma_{u_M}(0)\right)_i
    \end{bmatrix}.
\end{equation*}}
Letting $D := ((\gamma_{u_j}(T_j))-(\gamma_{u_j}(0))_i)_{i,j=1}^{n,M} $ it can be concluded that $ \delta^\top G_d = D$. Finally, the modes $\xi$ are given by $\xi = D V$ and the estimated open-loop model is given by
\begin{equation}\label{eq:convergent_closed_loop_model}
    \dot{x} \approx \hat{F}_{M}(x,u) = D V\Sigma W^\top \beta(x)\begin{pmatrix}
        1\\u
    \end{pmatrix} = D G_\beta^{+} \beta(x) \begin{pmatrix}
        1\\u
    \end{pmatrix}.
\end{equation}
The approximation $\hat{f}_M(x)$, an approximation of the drift dynamics, $f(x)$, is given by the first column of $D G_\beta^{+} \beta(x)$ and $\hat{g}_M(x)$, an approximation of the control-effectiveness matrix, $g(x)$, is given by the last $m$ columns of $D G_\beta^{+} \beta(x)$.

Since $P_{\beta} A_{f,g} P_{d}$ converges to $A_{f,g}$ in norm as $M\to\infty$, and since the space $F^2_{\tilde{\rho}_d}(\mathbb{R}^n)$ contains $h_{\mathrm{id},j}$ for $j=1,\cdots,n$, the following result is immediate.
\begin{corollary}\label{cor:uniform-convergence}
    Under the hypothesis of Theorem \ref{thm:norm-convergence}, for all $u\in\mathcal{U}$, where for all $t\in [0,T]$, $\Vert \begin{pmatrix}
       1 & u^\top(t) 
    \end{pmatrix}  \Vert <\overline{U}$, $\lim_{M\to\infty}\left(\sup_{x\in X}\left\Vert \hat{F}_{M}(x,u) - F(x,u) \right\Vert_{2}\right) = 0$.
\end{corollary}
\begin{proof}

   Since the space $F^2_{\tilde{\rho}_d}(\mathbb{R}^n)$ contains $h_{\mathrm{id},j}$ for $j=1,\cdots,n$, the functions $P_{\beta^j} A_{f,g} P_{d^j} h_{\mathrm{id},j} $ and $ A_{f,g} h_{\mathrm{id},j} $ exist as members of $H$. Since  $x\mapsto K(x,x) = \diag\begin{pmatrix}
       \exp\left(\frac{x^\top x}{\tilde{\rho}_1}\right)&\ldots&\exp\left(\frac{x^\top x}{\tilde{\rho}_{m+1}}\right)
   \end{pmatrix}$ is continuous and $X$ is compact, there exists a real number $\overline{K}$ such that $\sup_{x\in X}\left\Vert K(x,x)\right\Vert_{\mathcal{Y}}^{\mathcal{Y}} = \overline{K}$.

Let $Y(x)=\begin{pmatrix}
    f(x) &g(x)
\end{pmatrix}\in \mathbb{R}^{n \times (m+1)}$, and $\hat{Y}^M(x)=\begin{pmatrix}
    \hat{f}^M(x) &\hat{g}^M(x)\end{pmatrix}=DV\Sigma W^\top \beta(x) 
\in \mathbb{R}^{n \times (m+1)}$ is the identified system using the finite-rank representation with rank at-most $M$, and $Y_j(x)$ is the $j^{th}$ row of $Y(x)$ .
   
Theorem \ref{thm:norm-convergence} can then be used to conclude that for all $\epsilon > 0$, $j=1,\ldots,n$, there exists $M(j) \in \mathbb{N}$ such that for all $i\geq M(j)$, $\left\Vert \hat{Y}^i_{j} - Y_{j}\right\Vert_{H}^{2} \leq \frac{\epsilon^2}{\overline{U}^2 \overline{K}}$. Using the reproducing property, for $i\geq \overline{M} \coloneqq \max_{j} M(j)$, 
\begin{multline*}
     \left((F(x,u))_j-(\hat{F}_M(x,u))_j\right)^2=\\
     \left\langle (\hat{Y}^i_{j}(x) - Y_{j}(x)),\begin{pmatrix}
         1 & u^\top
     \end{pmatrix} \right\rangle_{\mathcal{Y}}^{2}=\\
     \left\langle (\hat{Y}^i_{j} - Y_{j}), k_{x,\begin{pmatrix}
         1 & u^\top
     \end{pmatrix}} \right\rangle_{H}^{2} \leq\\
     \left\Vert \hat{Y}^i_{j} - Y_{j} \right\Vert_{H}^{2} \left\Vert k_{x,\begin{pmatrix}
         1 & u^\top
     \end{pmatrix}}\right\Vert_{H}^{2} =\\
    \left\Vert \hat{Y}^i_{j} - Y_{j} \right\Vert_{H}^{2}      \left\langle k_{x,\begin{pmatrix}
         1 & u^\top
     \end{pmatrix}}, k_{x,\begin{pmatrix}
         1 & u^\top
     \end{pmatrix}} \right\rangle_{H}=\\
       \left\Vert \hat{Y}^i_{j} - Y_{j} \right\Vert_{H}^{2}      \left\langle k_{x}\begin{pmatrix}
         1 & u^\top
     \end{pmatrix}, k_{x}\begin{pmatrix}
         1 & u^\top
     \end{pmatrix} \right\rangle_{H}=\\
       \left\Vert \hat{Y}^i_{j} - Y_{j} \right\Vert_{H}^{2}      \left\langle \begin{pmatrix}
         1 & u^\top
     \end{pmatrix}, k^*_{x}k_{x}\begin{pmatrix}
         1 & u^\top
     \end{pmatrix} \right\rangle_{\mathcal{Y}}=\\
     \left\Vert \hat{Y}^i_{j} - Y_{j} \right\Vert_{H}^{2}      \left\langle \begin{pmatrix}
         1 & u^\top
     \end{pmatrix}, K(x,x)\begin{pmatrix}
         1 & u^\top
     \end{pmatrix} \right\rangle_{\mathcal{Y}}=\\
     \left\Vert \hat{Y}^i_{j} - Y_{j} \right\Vert_{H}^{2} \begin{pmatrix}
         1 & u^\top
     \end{pmatrix} K(x,x)\begin{pmatrix}
         1 \\ u
     \end{pmatrix} \leq \\
     \left\Vert \hat{Y}^i_{j} - Y_{j} \right\Vert_{H}^{2} \overline{K}\left\Vert\begin{pmatrix}
         1 \\ u
     \end{pmatrix}\right\Vert^2 \leq \\
     \frac{\epsilon^2}{\overline{U}^2 \overline{K}} \overline{U}^2 \overline{K} = \epsilon^2\\
\end{multline*}
As a result, for all $\epsilon \geq 0$  there exists $\overline{M}$ such that for all $i\geq \overline{M}$,
\begin{multline*}
    \sup_{x\in X}\left\Vert \hat{F}_{M}(x,u) - F(x,u) \right\Vert_{2}= \\
    \sup_{x\in X} \left\Vert (\hat{Y}^{M}(x) - Y(x)) \begin{pmatrix}
        1\\u
    \end{pmatrix} \right\Vert_{2} \leq  \epsilon,
\end{multline*}
which completes the proof.
\end{proof}
The SCLDMD technique is summarized in Algorithm \ref{alg:SCLDMD}. The characterization $ \Gamma_{\gamma_{u_j}} = \int_0^{T_j} \tilde{K}\left(\cdot,\gamma_{u_j}(t)\right) \mathrm{d}t$ of occupation kernels is used on line \ref{line:psi}.
The inner product of two control occupation kernels is given in \cite{rosenfeld2024dynamicmodedecompositioncontrol} as
\begin{multline}
    \left\langle \Gamma_{\gamma_{u_i},u_i}, \Gamma_{\gamma_{u_j},u_j} \right\rangle_{H}\\
    =\int\limits_{0}^{T_j}\int\limits_{0}^{T_i} \begin{bmatrix}1 & u_i^{\top}(\tau)\end{bmatrix}  K\left(\gamma_{u_j}(t),\gamma_{u_i}(\tau)\right)\begin{bmatrix}
    1 \\ u_j(t)
    \end{bmatrix}\mathrm{d}\tau \mathrm{d}t.\label{eq:Control_occ_ker_gram_matrix}
\end{multline}
\begin{algorithm}
    \caption{\label{alg:SCLDMD}The SCLDMD algorithm}
    \begin{algorithmic}[1]
        \renewcommand{\algorithmicrequire}{\textbf{Input:}}
        \renewcommand{\algorithmicensure}{\textbf{Output:}}
        \REQUIRE Trajectories $\{\gamma_{u_i}\}_{i=1}^{M}$, a control signal $u$, a numerical integration procedure, and reproducing  kernels $\tilde{K}_d$ and $K$ of $\tilde{H}_d$ and $H$, respectively.
        \ENSURE $\{\xi_j,\sigma_j,\varphi_j,\phi_j\}_{j=1}^{M}$
        \STATE $G_\beta \leftarrow \left(\left\langle \Gamma_{\gamma_{u_i},u_i}, \Gamma_{\gamma_{u_j},u_j} \right\rangle_{H} \right)_{i,j=1}^M$ (see \eqref{eq:Control_occ_ker_gram_matrix})
        \STATE $ D \leftarrow \left( \left(\gamma_{u_j}(T_j)\right)_i - \left(\gamma_{u_j}(0)\right)_i\right)_{i,j=1}^{n,M}$
        \STATE $(W,\Sigma,V)\leftarrow$ SVD of $ G_\beta^{+} $
        \STATE $\xi \leftarrow DV$
        \STATE $\phi_j \leftarrow \sum_{i=1}^M (V)_{i,j} \left(K_d(\cdot,\gamma_{u_i}(T_i)) - K_d(\cdot,\gamma_{u_i}(0))\right)$
        \STATE $\psi_j \leftarrow \sum_{i=1}^M \int_0^{T_i} (W)_{i,j}\left[\begin{bmatrix} 1 & u_i(t)^{\top} \end{bmatrix} K_{\gamma_{u_i}\left(t\right)}\right]\left(\cdot\right) \mathrm{d}t$\label{line:psi}
        \RETURN $\{\xi_j,\sigma_j,\varphi_j,\phi_j\}_{j=1}^{M}$ 
    \end{algorithmic} 
\end{algorithm}

\section{Numerical Experiments}\label{sec:NumExp}
The purpose of the numerical experiment is to demonstrate the efficacy of the SCLDMD algorithm.

This experiment utilizes the nonlinear model of the Duffing oscillator given by 
\begin{equation}
\label{eq:Duffing_oscilator}
\dot{x}=f(x)+g(x)u=\begin{pmatrix} x_2\\ x_1-x_1^3\end{pmatrix}+\begin{pmatrix} 0\\ 2+\sin(x_1)\end{pmatrix}u,
\end{equation}
where $f(x)=\begin{pmatrix} x_2\\ x_1-x_1^3\end{pmatrix}$ is the drift function, $g(x)=\begin{pmatrix} 0\\ 2+\sin(x_1)\end{pmatrix}$ is the control effectiveness function, and $u$ is the controller. To approximate the system dynamics, 225 trajectories of the system are recorded, along with the corresponding control signals, starting from a grid of initial conditions, under a control signal that is composed of the sum of 15 sinusoidal signals with randomly generated frequencies between 1 and 3 radians per second, randomly generated phases between -1 and 1 radians, and randomly generated coefficients between -1 and 1.  The recorded trajectories and control signals, which are stored with a time step of 0.05 second and a duration of 1 second, are then utilized to construct the estimates $\hat{f}_M$ and $\hat{g}_M$ of $f$ and $g$, respectively. The grid of initial conditions is composed of 15 points in the $x_1$ coordinate and 15 points in the $x_2$ coordinate, equally spaced in the interval $[-3,3]$ in each coordinate. The exponential dot product kernel $\tilde{K}_d(x,y)=\exp(\frac{x^\top y}{\mu})$ is used to define 
$\tilde{H}_d$ with $\mu=11$, and the kernel operator $K$ is selected to be a diagonal matrix of exponential dot product kernels with the parameter $\mu_v=10$. Simpson's 1/3 rule is used for numerical integration.

\begin{figure}
   \centering

    \begin{tikzpicture}
    \begin{axis}[
        xlabel={Time (s)},
        legend pos = south west,
        enlarge y limits=0.1,
        enlarge x limits=0,
        height = 0.7\columnwidth,
        width = 1\columnwidth,
        label style={font=\scriptsize},
        tick label style={font=\scriptsize},
        ylabel shift = -8 pt,
    ]
        \addplot [thick, blue] table [x index=0, y index=1]{results/DuffingSCLDMDReconstruction.dat};
        \addplot [thick, red] table [x index=0, y index=2]{results/DuffingSCLDMDReconstruction.dat};
        \addplot [dashed, green] table [x index=0, y index=3]{results/DuffingSCLDMDReconstruction.dat};
        \addplot [dashed, black] table [x index=0, y index=4]{results/DuffingSCLDMDReconstruction.dat};
        \legend{$x_1(t)$, $x_2(t)$,$\hat{x}_1(t)$,$\hat{x}_2(t)$}

    \end{axis}
\end{tikzpicture}
    \caption{The responses of the true Duffing oscillator and the identified duffing oscillator to the input $u(t)=\sin(t) + \cos(2t)$.} \label{fig:duffing_response_approx_vs_actual}
\end{figure}
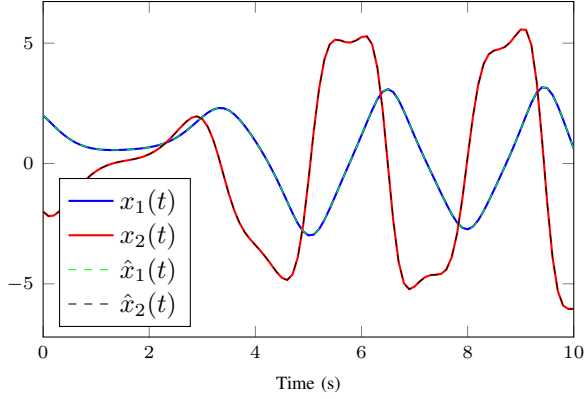

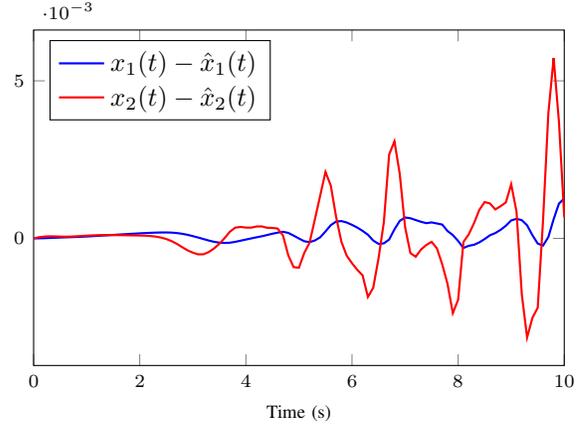
\begin{figure}
   \centering

    \begin{tikzpicture}
    \begin{axis}[
        xlabel={Time (s)},
        legend pos = north west,
        enlarge y limits=0.1,
        enlarge x limits=0,
        height = 0.7\columnwidth,
        width = 1\columnwidth,
        label style={font=\scriptsize},
        tick label style={font=\scriptsize},
        ylabel shift = -8 pt,
    ]
        \addplot [thick, blue] table [x index=0, y index=1]{results/DuffingSCLDMDError.dat};
        \addplot [thick, red] table [x index=0, y index=2]{results/DuffingSCLDMDError.dat};
        \legend{$x_1(t)-\hat{x}_1(t)$, $x_2(t)-\hat{x}_2(t)$}

    \end{axis}
\end{tikzpicture}
    \caption{The difference in responses of the true Duffing oscillator and the identified duffing oscillator to the input $u(t)=\sin(t) + \cos(2t)$.} \label{fig:duffing_error_approx_minus_actual}
\end{figure}

\begin{figure}[ht]
    \centering
    \begin{tikzpicture}
    \begin{axis}[
        width=0.55\columnwidth,
        colormap/viridis,
        xlabel={$x_1$},
        ylabel={$x_2$},
        zlabel={Error},
        title style={font=\scriptsize},
        title={$\tilde{f}$},
        zmax=0.0002,
        label style={font=\scriptsize},
        ylabel style={yshift=0.3cm,xshift=0.25cm},
        xlabel style={yshift=0.3cm,xshift=-0.25cm},
        tick label style={font=\scriptsize},
        view={-35}{25},
        enlarge y limits=0,
        enlarge x limits=0,
        enlarge z limits=0.01]
    \addplot3[surf, mesh/rows=9, shader=interp] table [] {results/DuffingVectorFieldfError.dat};
    \end{axis}
    \end{tikzpicture}\hspace{-0.5em}
    \begin{tikzpicture}
    \begin{axis}[
        width=0.55\columnwidth,
        colormap/viridis,
        xlabel={$x_1$},
        ylabel={$x_2$},
        zlabel={Error},
        title style={font=\scriptsize},
        title={$\tilde{g}$},
        zmax=0.0002,
        label style={font=\scriptsize},
        ylabel style={yshift=0.3cm,xshift=0.25cm},
        xlabel style={yshift=0.3cm,xshift=-0.25cm},
        tick label style={font=\scriptsize},
        view={-35}{25},
        enlarge y limits=0,
        enlarge x limits=0,
        enlarge z limits=0.01]
    \addplot3[surf, mesh/rows=9, shader=interp] table [] {results/DuffingVectorFieldgError.dat};
    \end{axis}
    \end{tikzpicture}
    \caption{Error $ \left\Vert f(x) - \hat{f}(x) \right\Vert$ as a function of $x$ (left) and error $ \left\Vert g(x) - \hat{g}(x) \right\Vert$ as a function of $x$ (right).}
    \label{fig:vector_field_error}
\end{figure}
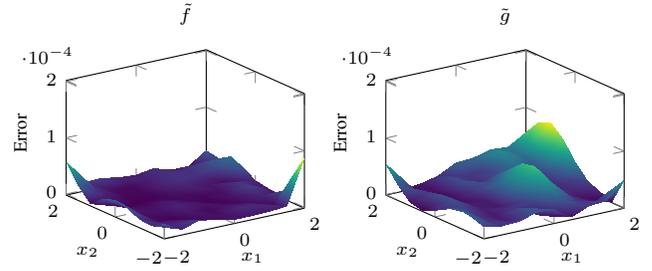

Figure \ref{fig:duffing_response_approx_vs_actual} shows the responses of the true Duffing oscillator from (\ref{eq:Duffing_oscilator}) and of the Duffing oscillator identified using SCLDMD, starting from the initial condition of $[2,-2]$ under the input $u(t)=\sin(t)+\cos(2t)$ for 10 seconds. The states of the identified model, $\hat{x}_1$ and $\hat{x}_2$, track closely the states $x_1$ and $x_2$ of (\ref{eq:Duffing_oscilator}). Figure \ref{fig:duffing_error_approx_minus_actual} shows the difference between the trajectories of the actual system and the  identified model, which is in the order of $10^{-3}$. The maximum error percentage for $\hat{x}_2$, which is the higher of the two, is 0.001\%.

Figure \ref{fig:vector_field_error} shows the vector field estimation errors for both the drift dynamics, $\tilde{f}$, and the control effectiveness, $\tilde{g}$, as functions of $x_1$ and $x_2$ in the domain $(x_1,x_2) \in [-2,2] \times  [-2,2]$. Both errors, $\tilde{f}$ and $\tilde{g}$ are of the order $10^{-4}$ in the above mentioned domain, where the error increases near the edges of the domain. Although $g$ from (\ref{eq:Duffing_oscilator}) is only a function of $x_1$, in Figure \ref{fig:vector_field_error} it can be seen that $\hat{g}$ changes with $x_2$. This is expected since the estimate $\hat{g}$ obtained through SCLDMD is a function of both states, given that the SCLDMD algorithm is not informed of the state dependencies of the control effectiveness. Exploring the incorporation of partial knowledge of the system to improve SCLDMD is part of future work.

\section{Conclusion} \label{sec:concl}
This paper introduces a novel approach towards the construction of a finite-rank representation of the control Liouville operator. New results on the construction of singular values and functions of the finite rank operator using singular values and vectors of a matrix representation are also obtained. Once the singular values and functions are at hand, the  drift dynamics and the control effectiveness can be approximated, which facilitates the prediction of the states of the dynamical system in response to an open-loop control signal. Moreover, the finite-rank representation of the control Liouville operator is shown to be convergent (in norm) to the true control Liouville operator, provided sufficiently rich data are available.

A numerical experiment using the Duffing oscillator is presented to demonstrate the performance of the developed technique. Analysis of the effects of integration error and measurement noise on the SCLDMD technique, as well as incorporating partial knowledge of the system to improve the identification technique are part of future work.

\bibliographystyle{IEEEtran}
\bibliography{scc,sccmaster,scctemp,refrences}
\end{document}